\documentclass[a4paper,UKenglish,cleveref, autoref]{lipics-v2019}

\usepackage{amsmath}
\usepackage{amsfonts}
\usepackage{amsthm}
\usepackage{amssymb}
\usepackage{dsfont}
\usepackage{hyperref}
\usepackage{graphicx}
\usepackage{color}
\usepackage{dcolumn}
\usepackage{bm}
\usepackage{physics}

\nolinenumbers

\newcommand{\enorm}[1]{\norm{#1}_{\mathrm{2}}}      
\newcommand{\snorm}[1]{\norm{#1}_{\mathrm {\infty}}}    

\newcommand{\set}[1]{{\left\{#1\right\}}}    

\newcommand{\ve}[1]{\mathbf{#1}}

\newcommand{\complex}{{\mathbb C}}
\newcommand{\reals}{{\mathbb R}}

\newcommand{\F}{F_{\alpha,\beta,\gamma}}
\newcommand{\G}{\Gamma_{\alpha,\beta,\gamma}}
\newcommand{\class}[1]{\textup{#1}}

\newcommand{\maxsat}{Max-SAT}

\def\ket#1{ | #1 \rangle}
\def\bra#1{{\langle #1 | }}

\newcommand{\spa}[1]{\mathcal{#1}}

\newcommand{\herm}[1]{\mathcal{H}(#1)}
\newcommand{\dens}[1]{\mathcal{D}(#1)}

\newcommand{\sep}{\rm SEP}

\newcommand{\opt}{{\rm OPT}}
\newcommand{\optprod}{{\rm OPT_{prod}}}

\newcommand{\rhoprod}{\rho_{\rm prod}}

\newcommand{\lmax}{\lambda_{\max}}
\newcommand{\lmin}{\lambda_{\min}}

\makeatletter
\newtheorem*{rep@theorem}{\rep@title}
\newcommand{\newreptheorem}[2]{%
\newenvironment{rep#1}[1]{%
 \def\rep@title{#2 \ref{##1}}%
 \begin{rep@theorem}}%
 {\end{rep@theorem}}}
\makeatother

\newreptheorem{theorem}{Theorem}

\title{Almost optimal classical approximation algorithms for a quantum generalization of Max-Cut}

\titlerunning{Approximation algorithms for a quantum generalization of Max-Cut}

\author{Sevag Gharibian}{University of Paderborn, Germany \and Virginia Commonwealth University, Richmond, VA, USA }{sevag.gharibian@upb.de}{}{NSF grants CCF-1526189 and CCF-1617710.}
\author{Ojas Parekh}{Sandia National Laboratories, Albuquerque, New Mexico, USA}{odparek@sandia.gov}{}{Laboratory Directed Research and Development program at Sandia National Laboratories, a multimission laboratory managed and operated by National Technology and Engineering Solutions of Sandia, LLC., a wholly owned subsidiary of Honeywell International, Inc., for the U.S. Department of Energy's National Nuclear Security Administration under contract DE-NA-0003525. Also supported by the U.S. Department of Energy, Office of Science, Office of Advanced Scientific Computing Research, Quantum Algorithms Teams program.}

\authorrunning{S.\,Gharibian and O.\, Parekh}

\Copyright{Sevag Gharibian and Ojas Parekh}



\category{}

\relatedversion{}

\supplement{}

\acknowledgements{We thank David Gosset and Mark Wilde for helpful discussions, and an anonymous referee for catching a technical error in an earlier version of this draft.}

\hideLIPIcs  

\EventEditors{Dimitris Achlioptas and L\'{a}szl\'{o} A. V\'{e}gh}
\EventNoEds{2}
\EventLongTitle{Approximation, Randomization, and Combinatorial Optimization. Algorithms and Techniques (APPROX/RANDOM 2019)}
\EventShortTitle{\mbox{\scriptsize APPROX/RANDOM 2019}}
\EventAcronym{APPROX/RANDOM}
\EventYear{2019}
\EventDate{September 20--22, 2019}
\EventLocation{Massachusetts Institute of Technology, Cambridge, MA, USA}
\EventLogo{}
\SeriesVolume{145}
\ArticleNo{31}

\begin{document}

\maketitle

\begin{abstract}
    Approximation algorithms for constraint satisfaction problems (CSPs) are a central direction of study in theoretical computer science. In this work, we study classical product state approximation algorithms for a physically motivated quantum generalization of Max-Cut, known as the quantum Heisenberg model. This model is notoriously difficult to solve exactly, even on bipartite graphs, in stark contrast to the classical setting of Max-Cut. Here we show, for any interaction graph, how to classically and efficiently obtain approximation ratios $0.649$ (anti-ferromagnetic XY model) and $0.498$ (anti-ferromagnetic Heisenberg XYZ model). These are almost optimal; we show that the best possible ratios achievable by a product state for these models is $2/3$ and $1/2$, respectively.
\end{abstract}

\section{Introduction}\label{scn:intro}

The study of approximation algorithms for NP-complete problems is a central area of research in theoretical computer science (see, e.g.,~\cite{H97,V01}). Indeed, the field has seen breakthroughs such as the celebrated Goemans-Williamson~\cite{GW95} 0.878-approximation algorithm for Max-Cut, and the PCP theorem~\cite{AS98,ALMSS98}, which yielded a general framework for showing hardness of approximation results. Here, an approximation algorithm $A$ with ratio $0<r<1$ is defined as follows: Given an instance $\Pi$ of a maximization problem with optimal value $\opt$, $A$ runs in polynomial time and outputs a value $\widetilde{\opt}$ satisfying $r\opt\leq \widetilde{\opt}\leq \opt$. Focal points of study in approximation algorithms are Boolean constraint satisfaction problems (CSPs) such as \maxsat\ and Max-Cut, in which one is roughly given a set of local constraints acting on $k\in O(1)$ bits each (out of a total of $n$ bits), and asked to compute the largest number of constraints which are simultaneously satisfiable.

In the quantum setting, CSPs are naturally generalized by the \emph{$k$-local Hamiltonian problem ($k$-LH)}~\cite{KSV02}. In the latter, one is given as input an exponentially large (in the number of qubits, $n$) Hermitian matrix $H$ known as a \emph{local Hamiltonian}, which has a succinct description in terms of local ``quantum clauses.'' The goal is to estimate the smallest eigenvalue of $H$, $\lmin(H)$, i.e. the \emph{ground state energy} of $H$. Slightly more formally, a $k$-local Hamiltonian $H=\sum_{S\subseteq[n]}H_S$ acts on $n$ qubits in total, with each local quantum ``clause'' $H_S$ acting on a constant number $k$ of qubits denoted by subset $S\subseteq[n]$ with $\abs{S}=k$. (Thus, each $H_S$ is a $2^k\times 2^k$ Hermitian matrix. Note that formally, $H_S$ implicitly denotes operator $I_{[n]\setminus S}\otimes H_S$; this ensures dimensions match in the sum over clauses.)  Quantum CSPs in which the matrices $H_S$ are diagonal correspond to classical CSPs.

The problem $k$-LH is not only physically motivated (it is the problem of estimating the energy of a quantum many-body system when cooled to near absolute zero), but also complexity theoretically --- it was the first known QMA-complete problem~\cite{KSV02}, where Quantum Merlin Arthur (QMA) is the quantum analogue of NP. As such, $k$-LH has been a central problem of study in the field of {Quantum Hamiltonian Complexity} (see, e.g.~\cite{O11,GHLS14} for surveys), which (among other aims) uses tools from complexity theory to uncover the limits and structure of physical systems in nature. In recent years, this interdisciplinary research has led to a growing body of work on classical \emph{approximation algorithms} for $k$-LH. It is this direction which we pursue in this paper.

\subsection{Product state algorithms and previous work} 

We begin by reviewing previous work on approximation algorithms for $k$-LH.

\paragraph*{Mean-field or product-state algorithms.} All known classical approximation algorithms for $k$-LH fall under the category of \emph{mean-field} or \emph{product-state} algorithms. Here, the issue is that the optimal solution to a $k$-LH instance may be an exponentially large quantum state $\ket{\psi}\in\complex^{2^n}$ (which would be the \emph{ground state} or eigenvector of $H$ corresponding to its ground state energy, $\lmin(H)$). Any classical algorithm for approximating $k$-LH must hence presumably pick a reasonable succinctly representable class of quantum states to optimize over; the simplest such class is the set of $n$-qubit \emph{product states}. A product state is the quantum analogue of a product distribution --- the entire $2^n$-dimensional vector $\ket{\psi}$ is fully specified by locally giving an assignment $\ket{\psi_i}\in\complex^2$ to each qubit $i$, i.e. $\ket{\psi}=\ket{\psi_1}\otimes\cdots\otimes\ket{\psi_n}$.

\emph{Remark.} It is crucial to note that even though product states are not entangled, they nevertheless generalize classical bit string assignments, and are thus \emph{NP-hard to optimize over} in the worst case. Thus, even with this simplest ansatz of product states, approximating $k$-LH is highly non-trivial.

\paragraph*{Previous work for QMA-complete models.} We now outline the known approximation algorithms for $k$-LH, which are all mean-field algorithms. The first such work was due to Bansal, Bravyi, and Terhal~\cite{BBT09}, who gave a classical polynomial-time approximation scheme (PTAS) for $k$-LH on bounded degree planar graphs. Next, Gharibian and Kempe~\cite{GK12} gave a PTAS for computing {product-state} solutions to dense CSPs, and showed their algorithm yielded a $d^{1-k}$ approximation for dense $k$-LH on local $d$-dimensional systems. Brand\~{a}o and Harrow~\cite{BH13} then gave PTAS-es for $k$-LH in three settings: Planar, dense, and low threshold rank graphs.  Most recently, Bravyi, Gosset, K\"onig, and Temme~\cite{BGKT18} gave a $O(\log n)$-approximation algorithm for traceless 2-local Hamiltonians. As we shall see, this last work may be viewed as complementary to ours (and indeed, the techniques used are similar, although independently developed) --- the algorithm of~\cite{BGKT18} is more general than ours (applies to all traceless Hamiltonians) but has a non-constant approximation ratio ($O(\log n)$ ratio). We take the complementary route: We study a more specific model, the central quantum Heisenberg model, but in return are able to achieve substantially stronger \emph{constant} approximation ratios. Finally, Lee and Hallgren~\cite{HL19} obtain a non-trivial constant-factor approximation algorithm for $2$-LH when each clause is positive semi-definite. We remark that with the exception of ~\cite{BBT09}, all of these works are based on semidefinite programs (SDP).

\paragraph*{Previous work for Hamiltonians of ``intermediate'' complexity.} For completeness, we also note that Bravyi~\cite{B14} and Bravyi and Gosset~\cite{BG16} showed fully polynomial randomized approximation schemes (FPRAS) for approximating the partition function\footnote{The ability to compute the partition function allows one in turn to solve $k$-LH.} of certain {ferromagnetic} models, such as the ferromagnetic transverse field Ising model (ferromagnetic TIM). In general, the TIM problem is StoqMA-complete, as shown by Bravyi and Hastings~\cite{BH14}. Here, $\class{MA}\subseteq\class{StoqMA}\subseteq\class{QMA}$, and it is generally believed $\class{StoqMA}$ is strictly smaller than QMA (the former is in the Polynomial-time Hierarchy, whereas the latter is believed not to be). Thus, such models may be thought of as being of ``intermediate'' complexity.

\paragraph*{Brief note on the quantum PCP theorem.} An advantage of any mean-field classical approximation algorithm for $k$-LH is that it yields negative progress on the central open question: Does a quantum PCP theorem\footnote{Recently, the ``entangled non-local games'' version of the PCP theorem has been established under randomized reductions~\cite{NV18}. The ``hardness of approximation'' version involving approximating ground state energies of local Hamiltonians, however, which is relevant to this work, remains open.} hold~\cite{AALV09,AAV13}? This is because such algorithms show that a \emph{classical} (i.e. NP) witness suffices to attain certain approximation ratios for $k$-LH. Thus, unless $\class{NP}=\class{QMA}$ (which is believed highly unlikely), a quantum PCP theorem for $k$-LH with the same approximation ratios cannot hold.




\subsection{Our results}
We give classical approximation algorithms for a maximization version of the fundamental \emph{quantum Heisenberg model}, which can be thought of as a family of Hamiltonians generalizing the NP-complete Max-Cut problem.

\paragraph*{Maximization versus minimization.} For clarity, we study the natural \emph{maximization} variant of $k$-LH, in which one is given $H$ and asked to estimate its \emph{largest} eigenvalue $\lmax(H)$. We study this variant for two reasons (see also~\cite{GK11}): First, in the minimization setting, if $\lmin(H)=0$, the notion of an approximation ratio is not well-defined, and second, the maximization setting allows us to naturally align with classical approximation algorithms for CSPs such as Max-Cut. We remark that in the exact setting, computing $\lmin(H)$ is equivalent in complexity to computing $\lmax(H)$ since $\lmin(H)=\lmax(-H)$ --- thus, both maximization and minimization variants of $k$-LH are QMA-complete.  More precisely, if $H$ is a Hamiltonian corresponding to an instance of the (anti-ferromagnetic) quantum Heisenberg model, then we approximate the instance $\lmax(mI-H)$, where $m$ is the number of clauses. In terms of approximability, the complexity of both models need not coincide.  An appropriate classical analogy is the relationship of the Ising problem on graphs, $\min_{z_i \in \{-1,1\}} \sum_{ij \in E} z_i z_j$, for which an $O(\log n)$-approximation is the best known (see, e.g.,~\cite{CW04}) and the Max-Cut problem, $\max_{z_i \in \{-1,1\}} \sum_{ij \in E} (1-z_i z_j)/2$, for which the Goemans-Williamson $0.878$-approximation is known.  These problems are equivalent from an exact optimization perspective.  From an approximation perspective, the standard quantum Heisenberg model is a generalization of the Ising problem, while the problem we study is a generalization of Max-Cut (see Appendix~\ref{app:MC} for details).  We note that Bravyi et al.'s $O(\log n)$-approximation for traceless 2-local Hamiltonians~\cite{BGKT18} includes the standard quantum Heisenberg model as a special case.

\paragraph*{The quantum Heisenberg model.} The Heisenberg model is fundamental to the study of magnetism, and has received attention for at least almost a century now~(e.g. the well-known Bethe ansatz of 1931~\cite{B31}). It is a family of $2$-local Hamiltonians, defined in this paper as having constraints $H_{ij}$ acting on qubits $i$ and $j$ of the form (see Section~\ref{scn:prelims} for formal definitions):
\[
    H_{ij}=I-\alpha X_i\otimes X_j-\beta Y_i\otimes Y_j -\gamma Z_i\otimes Z_j,
\]
for Pauli matrices $X,Y,Z$, and where $X_i$ indicates $X$ acts on qubits $i$. (Recall we study \emph{maximization}, i.e. estimating $\lmax(H)$.) Three important  well-known special cases of this model are: (1) the Max-Cut problem ($\alpha=\beta=0$, $\gamma=1$) (in Appendix~\ref{app:MC}, we sketch why this case indeed captures Max-Cut), (2) the (anti-ferromagnetic) XY model ($\alpha=\beta=1$, $\gamma=0$), and (3) the (anti-ferromagnetic) Heisenberg model ($\alpha=\beta=\gamma=1$), which we also refer to as the anti-ferromagnet. The latter, for example, is notoriously difficult to solve \emph{even on bipartite graphs}, in contrast to Max-Cut. The only solutions for the anti-ferromagnet we are aware of is on the 1D chain~\cite{B31} and on the complete graph (see, e.g.,~\cite{CM16}). This notoriety is well-deserved --- when non-negative polynomial-size weights are allowed on each constraint, both the XY model and anti-ferromagnet are QMA-hard~\cite{CM16,PM15}.

In this paper, we first show (Section~\ref{sscn:tight}) how to approximate the XY model and anti-ferromagnet almost optimally. The following is an informal statement (see Theorem~\ref{thm:alg1} for a formal statement).
\begin{theorem}\label{thm:temp}
    Let $\alpha,\beta,\gamma\in\set{0,1}$. Then, there exists a randomized, polynomial time classical algorithm for the quantum Heisenberg model which outputs a product state solution with ratio at least:
    \begin{itemize}
        \item $0.878$ if $\alpha+\beta+\gamma=1$ (equivalent to Max-Cut),
        \item $0.649$ if $\alpha+\beta+\gamma=2$ (equivalent to the XY model),
        \item $0.498$ if $\alpha+\beta+\gamma=3$ (anti-ferromagnet).
    \end{itemize}
\end{theorem}


We then show in Corollary~\ref{cor:ratio2} that these ratios are almost optimal, in the sense that the \emph{best} approximation ratios possible for a product state solution (whether efficiently attainable or not) to the XY model and anti-ferromagnet are at most $2/3$ and $1/2$, respectively. It should be noted that, in contrast, the naive ``random assignment'' strategy (i.e. choose the maximally mixed state $I/2^n$ as the assignment) yields ratios of only $1/3$ and $1/4$ for the XY model and anti-ferromagnet, respectively.

Next, in Section~\ref{sscn:beyondH} we give two ways in which our algorithm (or a variant of it) can be applied to a broader class of Hamiltonians:
\begin{itemize}
    \item Section~\ref{ssscn:varyweights} shows how to relax the constraint that $\alpha,\beta,\gamma\in\set{0,1}$. Specifically, we allow a different set of parameters $\alpha_{ij},\beta_{ij},\gamma_{ij}\in[-1,1]$ for each edge $(i,j)\in E$. In return for this generality, the approximation ratios we obtain are slightly weaker.
    \item Section~\ref{ssscn:localu} uses a trick from entanglement theory~\cite{HH96,HH96_2} to characterize the class of models which can be reduced to the Heisenberg model via application of local unitaries, and hence to which our algorithms apply.
\end{itemize}

\subsection{Techniques}

Our algorithms are based on semidefinite programming (SDP), and in particular use the first level of a non-commutative generalization of the Lasserre SDP hierarchy. Similar generalizations have been used previously in~\cite{BH13,BGKT18}. Note that a key difference between our approach and the previous SDP-based works of~\cite{GK11,BH13} is that the SDPs we derive are relaxations not just of the best attainable \emph{product state} objective function value, but rather of the true optimal value $\lmax(H)$ itself. This is why the ratios we obtain in Theorem~\ref{thm:alg1} can be close to optimal for a product state ansatz.  We note that a simple modification of our SDP relaxation does give an upper bound on the NP-hard problem of finding the best product-state solution; our techniques can be used to yield classical approximation algorithms for this problem as well.

\subsection{Open questions} Many questions in the study of approximability in the quantum setting remain open. For example, what are the best achievable approximation ratios classically for the Heisenberg model, and do hardness of approximation results based on the unique games conjecture yield tight bounds as they do for Max-Cut and related classical CSPs? Can tight ratios of $2/3$ and $1/2$ be obtained for the XY model and anti-ferromagnet, respectively? Are there constant-factor approximation algorithms for general $k$-LH (recall~\cite{BGKT18} give $O(\log n)$ approximations for traceless 2-local Hamiltonians)? How well can one approximate ``intermediate'' Hamiltonian models such as the \emph{anti-ferromagnetic} TIM (recall~\cite{B14,BG16} approximate the ferromagnetic TIM)? Can one optimize approximately over more general ansatzes than mean-field/product states, such as tensor network states? Can quantum approximation algorithms \emph{provably} outperform the best classical approximation algorithms? Finally, does a quantum PCP theorem (in the sense of ``hardness of approximation for quantum CSPs'') hold? It is hoped that the current paper will act as a step towards resolutions for some of these problems.

\subsection{Organization} In Section~\ref{scn:prelims}, we give definitions and preliminaries. Section~\ref{sscn:ubounds} gives upper bounds on the power of the mean-field ansatz. Section~\ref{sscn:tight} gives our approximation algorithms. Certain technical proofs are deferred to Appendix~\ref{app:sec2proof}. Some background in basic quantum information is assumed; see, for example, Nielsen and Chuang~\cite{NC00} for a standard reference.

\section{Preliminaries}\label{scn:prelims}

\subsection{Notation} Let $[n]:=\set{1,\ldots, n}$. The sets $\herm{\spa{X}}$ and $\dens{\spa{X}}$ denote the sets of Hermitian and density operators acting on complex Euclidean space $\spa{X}$. For $A,B\in\herm{\spa{X}}$, we say $A\succeq B$ if $A-B$ is positive semidefinite, i.e. $A-B\succeq 0$. The spectral/operator norm of $A$ is denoted $\snorm{A}=\trace(\sqrt{A^\dagger A})$.

\subsection{Physically motivated $2$-local Hamiltonians}

Let $G=(V,E)$ be a simple, undirected graph with $\abs{V}=n$ and $\abs{E}=m$. In this section, we study physically motivated $2$-local Hamiltonians $H$ based on the quantum Heisenberg model, $H=\sum_{(i,j)\in E}w_{ij}H_{ij}$ for $H_{ij}=\alpha X_i X_j+\beta Y_i Y_j + \gamma Z_i Z_j$ (more accurately, since we are in the setting of maximization, we use local terms as given in Equation~(\ref{eqn:heis})), where we consider $\alpha,\beta,\gamma\in\set{0,1}$ and $w_{ij}\geq 0$. This includes QMA-hard special cases such as the quantum Heisenberg anti-ferromagnet~\cite{CM16,PM15}. Here, $X$, $Y$, $Z$ are the Pauli matrices
\[
    X=\left(
      \begin{array}{cc}
        0 & 1 \\
        1 & 0 \\
      \end{array}
    \right),\quad
    Y=\left(
      \begin{array}{cc}
        0 & -i \\
        i & 0 \\
      \end{array}
    \right),    \quad
    Z=\left(
      \begin{array}{cc}
        1 & 0 \\
        0 & -1 \\
      \end{array}
    \right),
\]
and $X_i, Y_i, Z_i$ refer to the Pauli matrices acting on the $ith$ qubit (i.e., tensored with identity on all other qubits).

Specifically, we consider the equivalent (in the setting of exact computation) maximization variant where each local term is defined
\begin{equation}\label{eqn:heis}
    H_{ij}=I-\alpha X_i X_j-\beta Y_i Y_j -\gamma Z_i Z_j,
\end{equation}
and our goal is to estimate the \emph{largest} eigenvalue of $H=\sum_{(i,j)\in E}w_{ij}H_{ij}$ with $w_{ij}\geq 0$. This variant is clearly still QMA-hard, and includes as a special case, for example, the canonical NP-complete problem Max-Cut, obtained up to scaling by a constant factor of $2$) by setting $\alpha=0, \beta=0, \gamma=1$.

We now set definitions for the rest of this paper. Let $\F$ denote the set of all $H$ with (non-negative weighted) constraints of the form of Equation~(\ref{eqn:heis}), with parameters $\alpha,\beta,\gamma$ and on all interaction graphs $G$ (for all $n\geq 0$). For example, $F_{0,0,1}$ denotes the set of all possible Max-Cut instances with non-negative edge weights. In this paper, we refer to the family $F=\bigcup_{\alpha,\beta,\gamma\in\set{0,1}}\F$ as ``the {Heisenberg model}''. Let $\sep=\operatorname{conv}$$(\bigotimes_{i=1}^n \rho_i\mid \rho_i\in\dens{\complex^2})$ for $\operatorname{conv}(S)$ the convex hull of set $S$, i.e. $\sep$ is the set of fully separable quantum states on $n$ qubits.


\section{Upper bounds on product state ratios}\label{sscn:ubounds}

As quantum states on $n$ qubits generally require exponential space to represent, a classical approximation algorithm for estimating ground state energies must generally optimize over a restricted class of quantum states, or an \emph{ansatz}. Our ansatz in this section will be to optimize over $\sep$. To formalize this, we first define the notion of a product state ratio.

\paragraph*{Product state ratio.} Let $H\in\herm{(\complex^2)^{\otimes n}}$ be a Hermitian operator with largest eigenvalue $\opt(H)=\lmax(H)$, and let
\[
    \optprod(H):=\max_{\rho\in \sep}\trace(H\rho).
\]
  By convexity, the optimal $\rho$ here is a (pure) product state. The \emph{product state ratio} is defined as $\optprod(H)/\opt(H)$. For the Heisenberg model in particular, for any fixed $\alpha,\beta,\gamma\in\set{0,1}$, define
\[
    \G=\min_{H\in\F} \frac{\optprod(H)}{\opt(H)},
\]
the worst-case product state ratio over all Hamiltonians in $\F$.

By definition, $\G$ yields an upper bound on the best approximation ratio achievable by any approximation algorithm using a product state ansatz. It is thus crucial to understand $\G$, which we now do for the Heisenberg model. For this, we first give two lemmas which fully characterize the optimal product state ratio on a \emph{single} (unit weight) edge. Note the characterization we give is more general than how we defined the Heisenberg model here, in that it applies for any $\alpha,\beta,\gamma\in\reals$. (For clarity, the term proportional to the identity is omitted in Lemmas~\ref{l:optprod} and~\ref{l:opt} below, but is accounted for in the subsequent statement of Corollary~\ref{cor:ratio}.) The proofs of both lemmas are deferred to Appendix~\ref{app:sec2proof}.

\begin{lemma}\label{l:optprod}
    Let $H=\alpha X\otimes X + \beta Y\otimes Y+\gamma Z\otimes Z$ for $\alpha,\beta,\gamma\in\reals$. Then $\optprod(H)=\snorm{(\alpha,\beta,\gamma)}$.
\end{lemma}

\begin{lemma}\label{l:opt}
    Let $H=\alpha X\otimes X + \beta Y\otimes Y+\gamma Z\otimes Z$ for $\alpha,\beta,\gamma\in\reals$. Then
    \[
        \opt(H)=\max(\abs{\alpha-\beta}+\gamma,\abs{\alpha+\beta}-\gamma).
    \]
\end{lemma}

The following corollary now follows essentially immediately by applying Lemmas~\ref{l:optprod} and~\ref{l:opt} to a single unit weight edge of the form in Equation~\ref{eqn:heis} (i.e. with an identity term).
\begin{corollary}\label{cor:ratio}
        For any $\alpha,\beta,\gamma\in\reals$,
        \[
            \G\leq \frac{1+\max(\abs{\alpha},\abs{\beta},\abs{\gamma})}{1+\max(\abs{\alpha-\beta}-\gamma,\abs{\alpha+\beta}+\gamma)}.
        \]
\end{corollary}
\begin{proof}
   Combine Lemmas~\ref{l:optprod} and~\ref{l:opt} with the following additional observation: The values $\alpha,\beta,\gamma$, as defined for $\F$, should be interpreted as $-\alpha,-\beta,-\gamma$ for Lemmas~\ref{l:optprod} and~\ref{l:opt} due to how Equation~\ref{eqn:heis} is stated. As a result, the positions of the $\gamma$ and $-\gamma$ terms are swapped in the result of Lemma~\ref{l:opt}.
\end{proof}

We thus have the following for the special case of the Heisenberg model we consider here (i.e. $\alpha,\beta,\gamma\in\set{0,1}$).
\begin{corollary}\label{cor:ratio2}
    For any $\alpha,\beta,\gamma\in\set{0,1}$, if:
    \begin{itemize}
        \item $\alpha+\beta+\gamma=1$, then $\G=1$.
        \item $\alpha+\beta+\gamma=2$, then $\G\leq 2/3$.
        \item $\alpha+\beta+\gamma=3$, then $\G\leq 1/2$.
    \end{itemize}
    \begin{proof}
        When $\alpha,\beta,\gamma\in\set{0,1}$, the bound of Corollary~\ref{cor:ratio} simplifies to
        \[
            \G\leq \frac{2}{1+\alpha+\beta+\gamma},
        \]
        from which the upper bounds claimed follow. The matching lower bound for $\alpha+\beta+\gamma=1$ is obtained since $H$ can be mapped via local Pauli gates to $H'\in F_{0,0,1}$, i.e. $H'$ is diagonal in the standard basis. Thus, product states are optimal in this case. For example, applying local Hadamard gates to each qubit maps any $H\in F_{1,0,0}$ to $H'\in F_{0,0,1}$. (A matching lower bound can also be obtained for $\alpha+\beta+\gamma=3$ by observing that $H_{ij}\succeq 0$, and using the general result that any local Hamiltonian $H'$ (not necessarily from the Heisenberg model) with positive semidefinite constraints satisfies $\optprod(H')/\opt(H')\geq 1/2$~\cite{GK11}. However, unlike Theorem~\ref{thm:alg1}, the lower bound of~\cite{GK11} is not known to be efficiently achievable.)
    \end{proof}
\end{corollary}

\section{Almost optimal product-state approximation algorithms}\label{sscn:tight}

In Section~\ref{sscn:ubounds}, we gave upper bounds on $\G$ for the Heisenberg model. In this section, we give almost matching algorithmic lower bounds on $\G$ when $\alpha+\beta+\gamma\in\set{2,3}$ (recall $\alpha+\beta+\gamma=1$ is equivalent to Max-Cut, and so $\G=1$). Specifically, we give an approximation algorithm which is almost optimal in the following sense: Given $H\in\F$, it outputs a product state $\rhoprod$ with approximation ratio at least $0.649$ and $0.498$ when $\alpha+\beta+\gamma$ equals $2$ and $3$, respectively, which by Corollary~\ref{cor:ratio2} almost matches the best possible mean-field ratios of $2/3$ and $1/2$, respectively.

\begin{theorem}\label{thm:alg1}
    Let $H\in \F$ for $\alpha,\beta,\gamma\in\set{0,1}$. There exists a randomized, polynomial-time algorithm which obtains approximation ratios at least $0.878$, $0.649$ or $0.498$,  when $\alpha+\beta+\gamma$ equals $1$, $2$ or $3$, respectively.
\end{theorem}
\begin{proof}
    Suppose $H$ has interaction graph $G=(V,E)$ for $\abs{V}=n$ and edge weights $w_{ij}\geq 0$ for $(i,j)\in E$. We first define a semidefinite programming (SDP) relaxation of $\opt(H)$ via the first level of the Lasserre hierarchy (see, e.g.,~\cite{BH13} for a similar exposition for the setting of low threshold rank graphs). We then show that applying a generalization of the Goemans-Williamson (GW)~\cite{GW95,BFV12} rounding scheme yields the desired result.

    \paragraph*{The SDP.} Each solution of the SDP relaxation will be a ``moment matrix'' $M\in\reals^{3n\times 3n}$, whose rows (resp., columns) are indexed by $2$-tuples $(i,k)\in [n]\times [3]$ (resp., $(j,l)\in [n]\times [3]$) such that ideally, $i$,$j$ denote qubits, and $k,j$ a choice of Pauli matrix from sequence $(\sigma_1,\sigma_2, \sigma_3)=(X,Y,Z)$. Under this interpretation, an ideal solution $M$ corresponds to a density matrix $\rho\in\dens{(\complex^2)^{\otimes n}}$, such that
    \begin{equation}\label{eqn:ideal}
        M(ik,jl)=\trace(\rho\sigma_k^i\sigma_l^j),
    \end{equation}
    where $\sigma_k^i$ corresponds to Pauli operator $\sigma_k$ applied to qubit $i$, i.e. implicitly we have $\sigma_k^i\otimes I_{[n]\setminus\set{i}}$.

    Let us remark about the assumption that $M$ is real. Note that for an ideal solution (i.e. as in Equation~(\ref{eqn:ideal})), $M$ is Hermitian. Indeed, for $i\neq j$, $M(ik,jl)=M(jl,ik)\in\reals$, since the Pauli terms act on different qubits and hence commute. (A similar argument holds for $i=j$ and $k=l$.) If, however, $i=j$ and $k\neq l$, then since the Pauli matrices anti-commute, we have $M(ik,jl)=-M(jl,ik)$, and indeed $M(ik,jl),M(jl,ik)\in\complex\setminus\reals$ (since, e.g., $XY=iZ$), implying $M(ik,jl)^*=M(jl,ik)$ (thus $M$ is Hermitian; here, $^*$ denotes complex conjugate). Note, however, that the case of $i=j$ and $k\neq l$ corresponds to \emph{linear} local terms, i.e. those of the form $\sigma_k^i$, and these are the only non-real entries of $M$. Since our objective function involves only quadratic local terms (i.e. $\sigma_k^i\sigma_l^j$ for $i\neq j$), we can hence eliminate entries of $M$ with $i=j$  and $k\neq l$ by replacing $M$ with moment matrix $M'=(M+M^*)/2$, which is real and matches $M$ on all entries with $i\neq j$ (as well as on $i=j$ and $k=l$).  The real symmetric matrix $M'$ is positive semidefinite if the Hermitian $M$ is, and $M'$ results in an equal objective value to that of $M$, hence the restriction to real moment matrices is without loss of generality.

    We have thus far described the ideal solutions, $M$. Next, we add constraints to the SDP to help enforce this ideal interpretation of $M$:
    \begin{enumerate}
        \item For all $i\in[n], k\in[3]$, set $M(ik,ik)=1$, since ideally $M(ik,ik)=\trace(\rho\sigma_k^i\sigma_k^i)=\trace(\rho)=1$.
        \item For all $i\in[n], k\neq l\in[3]$, set $M(ik,il)=-M(il,ik)$, since distinct Pauli matrices anti-commute.
        \item Set $M\succeq 0$. This is since, ideally, for all $s\in\reals^{3n}$, we have
        \begin{eqnarray}
            s^TMs=\sum_{ijkl}s_{ik}s_{jl}M(ik,jl)=\trace\left(\rho\left(\sum_{ik}s_{ik}\sigma_k^i\right)\left(\sum_{jl}s_{jl}\sigma_l^j\right)\right)=\trace(\rho S^2)\geq 0,
        \end{eqnarray}
        where $S:=\sum_{ik}s_{ik}\sigma_k^i$, and since $\rho,S^2\succeq 0$.
    \end{enumerate}
    Finally, the relaxed objective function is obtained by replacing each term $\trace(\rho\sigma_k^i\sigma_l^j)$ with $M(ik,jl)$. For example, the relaxed objective function for $F_{1,1,1}$ becomes $\sum_{(i,j)\in E}w_{ij}(1-M(i1,j1)-M(i2,j2)-M(i3,j3))$.

    Let us remark that our formulation is essentially the first level $s=1$ of the Lasserre SDP hierarchy. Higher levels $s>1$ are obtained by considering $s$-local terms for the moment matrices, i.e. $M(i_1k_1,\ldots,i_sk_s)=\trace(\rho \sigma_{k_1}^{i_1}\cdots \sigma_{k_s}^{i_s})$.

    \paragraph*{Rounding solutions to the SDP.} Given any solution $M$ to the SDP, we take the Cholesky decomposition of $M$ to obtain a set of vectors $v_{ik}\in\reals^{3n}$ for $i\in[n]$ and $k\in[3]$, such that $M(ik,jl)=v_{ik}^Tv_{jl}$. Since $M(ik,ik)=1$, each $v_{ik}$ is a unit vector. Now, our aim is to round $M$ to a product state solution $\rhoprod=\rho_1\otimes\cdots\otimes \rho_n$ on $n$ qubits. Thus, writing $\rho_i$ in terms of its Bloch vector $\rho_i=(I+r_{i1}X+r_{i2}Y+r_{i3}Z)/2$ each $v_{ik}$ should be thought of as a $3n$-dimensional relaxation of $r_{ik}\in\reals$. For any $v\in\reals^p$, $w\in\reals^q$, define operation
    \[
    v\circ w =
     \begin{cases}
       0 &\quad\text{if $v={0}$ and $w={0}$}\\
       v &\quad\text{if $v\neq {0}$ and $w={0}$} \\
       w &\quad\text{if $w\neq {0}$ and $v={0}$} \\
       (v^T,w^T)^T&\quad\text{otherwise,} \\
     \end{cases}
\]
where $(v^T,w^T)^T\in\reals^{p+q}$ denotes the concatenation of $v$ and $w$. Recalling that $H\in\F$ for $\alpha,\beta,\gamma\in\set{0,1}$, we now set
\[
    u_i:=(\alpha v_{i1})\circ (\beta v_{i2}) \circ (\gamma v_{i3})\in\reals^{(\alpha+\beta+\gamma)3n}.
\]
This yields first that $w_{ij}(1-u_i^Tu_j)$ equals the term in the relaxed SDP objective function for $H$ corresponding to edge $(i,j)\in E$. For example, if $H\in F_{1,1,0}$ (i.e. the local terms are $w_{ij}(I-X_iX_j-Y_iY_j)$), then $u_i\in\reals^{6n}$ and for edge $(i,j)\in E$ we have $M(i1,j1)+M(i2,j2)=u_i^Tu_j$. Second, we have $\enorm{u_i}=\sqrt{\alpha+\beta+\gamma}$.

    To obtain the desired claim, define now $x_i=u_i/\enorm{u_i}$. We use a generalization of the Goemans-Williamson (GW)~\cite{GW95} rounding procedure due to Bri\"{e}t, de Oliveira Filho and Vallentin~\cite{BFV12}. Specifically, we randomly round each $x_i\in\reals^{(\alpha+\beta+\gamma)3n}$ to a Bloch vector $y_i\in\reals^{\alpha+\beta+\gamma}$ as follows. Let $R$ be a random $(\alpha+\beta+\gamma)\times (\alpha+\beta+\gamma)3n$ matrix, each of whose entries is chosen independently from a standard normal distribution with mean 0 and variance 1. Then, for each $i$, set
    \[
        y_i=Rx_i/\enorm{Rx_i}\in\reals^{\alpha+\beta+\gamma}.
    \]
    We map this to a (pure) single-qubit state $\rho_i$ as follows. Let $I(k)$ be the index in sequence $(\alpha,\beta,\gamma)$ of the $k$th non-zero entry (if it exists), for $k\in\set{1,2,3}$. Then, set the $I(k)$-th Bloch vector entry of $\rho_i$ to $y_{i,k}$. For example, if $\alpha=\beta=\gamma=1$, this yields $\rho_i=(I+y_{i,1}X+y_{i,2}Y+y_{i,3}Z)/2$, if $\alpha=\beta=1$ and $\gamma=0$, this yields $\rho_i=(I+y_{i,1}X+y_{i,2}Y)/2$, and if $\alpha=\beta=0$ and $\gamma=1$, this yields $\rho_i=(I+y_{i,1}Z)/2$ (note the subscript $1$ in $y_{i,1}$). For ease of exposition, henceforth we refer to the Bloch vector for $\rho_i$ as $\ve{r_i}=(r_1,r_2,r_3)$, where the entries of $\ve{r_i}$ which are not set in the rounding scheme above being implicitly set to $0$. For example, $\ve{r_i}=(y_{i,1},y_{i,2},y_{i,3})$, $\ve{r_i}=(y_{i,1},y_{i,2},0)$, and $\ve{r_i}=(0,0,y_{i,1})$, respectively, in the examples above.

    \paragraph*{Approximation ratio.} To analyze the approximation ratio obtained, note that for edge $(i,j)\in E$, we have
    \begin{eqnarray*}
        w_{ij}\trace(H_{ij}\rhoprod)&=&\frac{w_{ij}}{4}\trace(H_{ij}(I+r_{i,1}X_i+r_{i,2}Y_i+r_{i,3}Z_i)(I+r_{j,1}X_j+r_{j,2}Y_j+r_{j,3}Z_j))\\
        &=&\frac{w_{ij}}{4}\trace((I-\alpha X_i X_j-\beta Y_i Y_j -\gamma Z_i Z_j)\cdot\\&&\hspace{11mm}(I+r_{i,1}X_i+r_{i,2}Y_i+r_{i,3}Z_i)(I+r_{j,1}X_j+r_{j,2}Y_j+r_{j,3}Z_j))\\
        &=&w_{ij}(1-\alpha r_{i,1}r_{j,1}-\beta r_{i,2}r_{j,2}-\gamma r_{i,3}r_{j,3})\\
        &=&w_{ij}(1-y_i^T y_j).
    \end{eqnarray*}
    On the other hand, recall the SDP obtains value $w_{ij}(1-u_i^Tu_j)$ on edge $(i,j)\in E$. For brevity, let $F[r,\, u^T v]$ denote the right hand side of Equation~\ref{eqn:BFV} (Lemma~\ref{l:BFV} in Appendix~\ref{app:lemmas}).  A direct application of Lemma~\ref{l:BFV} yields $\mathds{E}[y_i^T y_j]= F[\alpha+\beta+\gamma,\,x_i^T x_j]$. Then, by linearity of expectation, the expected approximation ratio is given by the expected ratio attained on each edge, which is
    \[
        \frac{1-\mathds{E}[y_i^T y_j]}{1-u_i^T u_j}=\frac{1-F[\alpha+\beta+\gamma,\,x_i^T x_j]}{1-u_i^T u_j}=\frac{1-F[\alpha+\beta+\gamma,\,t]}{1-(\alpha+\beta+\gamma)t},
    \]
     where we defined $t=x_i^T x_j$ (note the value of $F$ only depends on $t$; see Appendix~\ref{app:lemmas}). Numerically evaluating via Mathematica (see Appendix~\ref{app:lemmas} for Mathematica code)
     \[
        \min_{t\in[-1,1/(\alpha+\beta+\gamma))} \frac{1-F[t]}{1-(\alpha+\beta+\gamma)t},
     \]
     we obtain ratios of $0.878$ (for $\alpha+\beta+\gamma=1$), $0.649$ (for $\alpha+\beta+\gamma=2$), and $0.498$ (for $\alpha+\beta+\gamma=3$), respectively. Note we minimize over $t\in[-1,1/(\alpha+\beta+\gamma))$, since for $t\in[1/(\alpha+\beta+\gamma),1]$ the ratio can only be negative (the denominator is negative, and the numerator is in range $[0,2]$). This completes the proof.
\end{proof}

\subsection{Generalizations beyond the Heisenberg model}\label{sscn:beyondH}

We defined the Heisenberg model $\F$ in Section~\ref{scn:prelims} as having all constraints identical with some fixed $(\alpha,\beta,\gamma)\in\set{0,1}^3$. We now show how to extend the algorithm to two more general settings: The first will allow different choices of $\alpha_{ij},\beta_{ij},\gamma_{ij}\in[1,-1]$ on each edge (note the use of $[1,-1]$ instead of $\set{0,1}$), and the second will require that all constraints remain identical but in exchange allows new interaction terms beyond $XX$, $YY$, $ZZ$.

\subsubsection{Approximating Heisenberg models with varying Pauli weights}
\label{ssscn:varyweights}

The approximation algorithm developed in the previous section made critical use of the fact that $\alpha, \beta, \gamma \in \{0,1\}$ for our Heisenberg model $\F$.  Here, we generalize by allowing two relaxations, captured below in the form of constraints now allowed:
\begin{equation*}
    H_{ij}=w_{ij}(I-\alpha_{ij} X_i X_j-\beta_{ij} Y_i Y_j -\gamma_{ij} Z_i Z_j),
\end{equation*}
where $\alpha_{ij},\beta_{ij},\gamma_{ij} \in [-1,1]$. The two relaxations to note are (1) $\alpha_{ij},\beta_{ij},\gamma_{ij}\in[-1,1]$ instead of in $\set{0,1}$, and (2) each edge $(i,j)\in E$ may have a different choice of $\alpha_{ij},\beta_{ij},\gamma_{ij}$ In this setting, we shall use the same relaxation as Section~\ref{sscn:tight}, but utilize another rounding strategy. In exchange for the added generality, the approximation ratios obtained are slightly weaker than those of Section~\ref{sscn:tight}.

In the theorem below, for brevity we call the sets $\set{\alpha_{ij}},\set{\beta_{ij}},\set{\gamma_{ij}}$ \emph{parameter families}. We say a parameter family is \emph{non-zero} if at least one parameter in the family is non-zero, e.g. there exists $(i,j)\in E$ such that $\alpha_{ij}\neq 0$ for family $\set{\alpha_{ij}}$.

\begin{theorem}\label{thm:alg2}
    Let $H=\sum_{(i,j)\in E}H_{ij}$ be a $2$-local Hamiltonian on qubits with constraints \[
        H_{ij}=w_{ij}(I-\alpha_{ij} X_i X_j-\beta_{ij} Y_i Y_j -\gamma_{ij} Z_i Z_j),
    \]
    where $\alpha_{ij},\beta_{ij},\gamma_{ij} \in [-1,1]$ and $w_{ij}\in\reals^+$. There exists a randomized, polynomial-time algorithm which obtains approximation ratio at least $0.878$ (if precisely one parameter family is non-zero), $0.609$ (if precisely two parameter families are non-zero), and $0.462$ (if all three parameter families are non-zero).
\end{theorem}
\begin{proof} We begin by mapping $H$ to a ``canonical'' form.

\paragraph*{Setup in ``canonical'' form.} For now, assume $\alpha_{ij},\beta_{ij},\gamma_{ij} \not=0$ (later we will get improved ratios when some of these values are 0 for every $(i,j) \in E$).  Our first observation is that we may assume $\alpha_{ij},\beta_{ij},\gamma_{ij} \in \{-1,1\}$.  This is because any vector $(\alpha_{ij},\beta_{ij},\gamma_{ij}) \in [-1,1]^3$ is a convex combination of vectors with coordinates in $\{-1,+1\}$ (i.e. the former lies in the convex hull of discrete points $(x,y,z)\in\set{-1,1}^3$).  Thus any $H_{ij}$ of the above form may be expressed as convex combination,
\begin{equation}\label{eqn:convex}
    H_{ij} = \sum_{k=1}^4 w_{ij} \lambda_k (I-\alpha_{ij,k} X_i X_j-\beta_{ij,k} Y_i Y_j -\gamma_{ij,k} Z_i Z_j),
\end{equation}
with $\alpha_{ij,k},\beta_{ij,k},\gamma_{ij,k} \in \{-1,1\}$, and $\lambda_k\geq 0$ with $\sum_{k=1}^4\lambda_k=1$. Notes: (1) Since we allow multiple edges between $i$ and $j$, we may include an edge for each term of the convex combination. (2) That we require at most $4$ terms $\lambda_k$ follows from Carath\'{e}odory's theorem, which says that a point in $\reals^d$ in the convex hull of some set $P$ requires at most $d+1$ points of $P$ to express as a convex combination. (3) Our approximation ratio analysis below will again be via expectation per edge, which by linearity of expectation yields that no loss in approximation is incurred by writing our constraints as in Equation~(\ref{eqn:convex}).

\paragraph*{Rounding algorithm.} We employ the same moment SDP relaxation as in Section~\ref{sscn:tight}, and continue to use the terminology therein. Consider the vectors $v_{i1},v_{i2},v_{i3} \in \reals^{3n}$ corresponding to an optimal solution of the SDP relaxation.  The objective value of the relaxation is
$w_{\text{SDP}}:= \sum_{(i,j) \in E} w_{ij}(1 -\alpha_{ij} v_{i1}^T v_{j1} -\beta_{ij} v_{i2}^T v_{j2} -\gamma_{ij} v_{i3}^T v_{j3})$.
Now suppose, without loss of generality (any other ordering is handled analogously):
\begin{equation*}
-\sum_{(i,j) \in E} w_{ij}\gamma_{ij} v_{i3}^T v_{j3} \geq -\sum_{(i,j) \in E} w_{ij}\beta_{ij} v_{i2}^T v_{j2} \geq -\sum_{(i,j) \in E} w_{ij}\alpha_{ij} v_{i1}^T v_{j1},
\end{equation*}
so that
\begin{equation}\label{eq:SDP-obj-bound}
\sum_{(i,j) \in E} w_{ij}(1 -3\gamma_{ij} v_{i3}^T v_{j3}) \geq w_{\text{SDP}}.
\end{equation}
Recall that the $v_{i3}$ are unit vectors (since our SDP had constraint $M(ik,ik)=1$ for all $i,k$). Hence, we may view the $v_{i3}$ as a feasible solution for the Max-Cut SDP relaxation of Goemans and Williamson and consequently, use the same rounding algorithm~\cite{GW95}:
\begin{enumerate}
\item Select a random vector $r \in \reals^{3n}$ with each entry independently and normally distributed with mean 0 and variance 1.
\item Let $r_i=r^T v_{i3}/|r^T v_{i3}| \in \{-1,1\}$.
\item Output the product state, $\prod_i \frac{1}{2}(I+r_iZ_i)$.
\end{enumerate}
Note that since the assignment above is diagonal in the $Z$ basis (i.e. is a standard basis state), it lies in the null space of each $XX$ and $YY$ term of our Hamiltonian. Consequently, our expected objective value for this assignment on our Hamiltonian is
\[
    w_{\text{EXP}}:=\sum_{(i,j) \in E} \mathds{E}[ w_{ij}(1 -\gamma_{ij} r_i r_j)] = \sum_{(i,j) \in E} w_{ij}(1 -\gamma_{ij} 2\arcsin(v_{i3}^T v_{j3})/\pi),
\]
where the second equality follows by (1) linearity of expectation and (2) the standard analysis of the Goemans-Williamson algorithm~\cite{GW95}, which states that $E[r_ir_j]=2\arcsin(v_{i3}^Tv_{j3})/\pi$.

\paragraph*{Approximation ratio.} We conclude by bounding the expected approximation ratio, $w_{\text{EXP}}/w_{\text{SDP}}$.  As for the analysis of the algorithm from the previous section, we need only consider the worst-case behavior on any edge.  Using \eqref{eq:SDP-obj-bound}, this is:
   \[
        \min_{\gamma \in \{-1,1\},\,t\in[-1,1]: 3\gamma t <1} \frac{1 -\gamma 2\arcsin(t)/\pi}{1 -3\gamma t},
     \]
where $\gamma$ represents $\gamma_{ij}$, and $t$ represents $v_{i3}^T v_{j3}$.  Numerically, this yields a ratio of 0.462.  A similar analysis produces an approximation ratio of 0.609 for the case when either $\alpha_{ij} = 0$ for all $(i,j)\in  E$, $\beta_{ij} = 0$ for all $(i,j)\in  E$, or $\gamma_{ij} = 0$ for all $(i,j)\in  E$.  We recover the Goemans-Williamson 0.878-approximation in the case when two of these parameters are 0 for all $(i,j) \in E$.
\end{proof}

\subsubsection{Reductions via local unitaries}\label{ssscn:localu}

We now generalize the algorithm of Section~\ref{sscn:tight} in a different manner. Specifically, using a standard trick from entanglement theory (used also in~\cite{CM16} in a somewhat different manner), we may give an approximation-preserving reduction to the Heisenberg model in certain cases. Namely, recall that any two-qubit Hermitian operator $H_{ij}$ can be expanded in the Pauli basis as follows (sometimes known as the \emph{Fano form}~\cite{F83}), given by:
\begin{equation}\label{eqn:fano}
    H_{ij} = \kappa I+ \sum_{a=1}^3\sum_{b=1}^3 M_{ab}\sigma_a\otimes\sigma_b +\sum_{a=1}^3 r_a\sigma_a\otimes I + \sum_{b=1}^3s_b I\otimes \sigma_b,
\end{equation}
where $\kappa,M_{ab},r_a,s_b\in\reals$. The $3\times 3$ real matrix $M$, which has no particular structure in general (for example, it need not be diagonalizable), is called the \emph{correlation matrix} in entanglement theory.

\begin{theorem}\label{thm:broader}
    Let $H$ be a $2$-local Hamiltonian on $n$ qubits, and with directed interaction graph $G=(V,E)$, where $H=\sum_{(i,j)\in E}w_{ij}H_{ij}$ for non-negative real weights $w_{ij}$. Assume
    \begin{enumerate}
        \item all $H_{ij}$ are identical with $\kappa=r_1=r_2=r_3=s_1=s_2=s_3=0$, and
        \item the correlation matrix $M$ of $H_{ij}$ is an orthogonal projection (i.e. $M$ is symmetric with $M^2=M$).
    \end{enumerate}
    Then, there exists a randomized, polynomial-time algorithm which obtains approximation ratios at least $0.878$, $0.649$ or $0.498$,  when the rank of $M$ equals $1$, $2$ or $3$, respectively. Conversely, the best possible product-state ratio (not necessarily efficiently attainable) in each case is $1$, $2/3$, and $1/2$, respectively.
\end{theorem}
\begin{proof}
    We use the approach of~\cite{HH96,HH96_2} of simulating orthogonal rotations on $M$ via local unitary operations on $H_{ij}$. Namely, due to the surjective homomorphism from SU(2) to SO(3), if one wishes to map $M$ to $O_1MO_2^T$ for orthogonal matrices $O_1$ and $O_2$, there exist single-qubit unitaries $U$ and $V$ such that $U_i\otimes V_j H_{ij} U_i^\dagger\otimes V_j^\dagger$ has correlation matrix $O_1MO_2^T$. Since $M$ is symmetric, it is diagonalizable by an orthogonal matrix $O\in \reals^{3\times 3}$ (Corollary 2.5.14 of~\cite{HJ90}). Thus, there exists a single-qubit unitary $U$ such that $U_i\otimes U_j H_{ij} U_i^\dagger\otimes U_j^\dagger$ has a diagonal correlation matrix with eigenvalues from set $\set{0,1}$. Since all constraints $H_{ij}$ are identical, it follows that $U^{\otimes n} H (U^\dagger)^{\otimes n}$ is a Hamiltonian in family $\F$ for some $\alpha,\beta,\gamma\in\set{0,1}$. The algorithm of Theorem~\ref{thm:alg1} now yields the claimed lower bound on approximation. The claimed upper bound on approximation follows from Corollary~\ref{cor:ratio2}. In both cases, we are leveraging the fact that our reduction applies only single-qubit unitary operations, and hence perfectly preserves approximation ratios attained by product states.
\end{proof}
Note that Theorem~\ref{thm:broader} uses the algorithm of Section~\ref{sscn:tight}. If we are willing to obtain slightly worse approximation ratios, we can relax the second requirement of Theorem~\ref{thm:broader} by instead applying the algorithm of Section~\ref{ssscn:varyweights}.

\begin{theorem}\label{thm:broader2}
    Let $H$ be a $2$-local Hamiltonian on $n$ qubits, and with directed interaction graph $G=(V,E)$, where $H=\sum_{(i,j)\in E}w_{ij}H_{ij}$ for non-negative real weights $w_{ij}$. Assume
    \begin{enumerate}
        \item all $H_{ij}$ are identical with $\kappa=r_1=r_2=r_3=s_1=s_2=s_3=0$, and
        \item the correlation matrix $M$ of $H_{ij}$ is symmetric.
    \end{enumerate}
    Then, there exists a randomized, polynomial-time algorithm which obtains approximation ratios at least $0.878$, $0.609$ or $0.462$,  when the rank of $M$ equals $1$, $2$ or $3$, respectively. Conversely, the best possible product-state ratio (not necessarily efficiently attainable) in each case is $1$, $2/3$, and $1/2$, respectively.
\end{theorem}
\noindent The proof is identical to that of Theorem~\ref{thm:broader}, except we using the rounding algorithm of Section~\ref{ssscn:varyweights} instead; we hence omit the proof.

%

\bibliography{approxvar}

\begin{thebibliography}{10}

\bibitem{AALV09}
D.~Aharonov, I.~Arad, Z.~Landau, and U.~Vazirani.
\newblock The detectibility lemma and quantum gap amplification.
\newblock In {\em Proceedings of 41st ACM Symposium on Theory of Computing
  ({STOC} 2009)}, volume 287, pages 417--426, 2009.

\bibitem{AAV13}
Dorit Aharonov, Itai Arad, and Thomas Vidick.
\newblock Guest column: The quantum {PCP} conjecture.
\newblock {\em SIGACT News}, 44(2):47--79, June 2013.
\newblock URL: \url{http://doi.acm.org/10.1145/2491533.2491549}, \href
  {http://dx.doi.org/10.1145/2491533.2491549}
  {\path{doi:10.1145/2491533.2491549}}.

\bibitem{ALMSS98}
S.~Arora, C.~Lund, R.~Motwani, M.~Sudan, and M.~Szegedy.
\newblock Proof verification and the hardness of approximation problems.
\newblock {\em Journal of the ACM}, 45(3):501--555, 1998.
\newblock Prelim. version FOCS '92.

\bibitem{AS98}
S.~Arora and S.~Safra.
\newblock Probabilistic checking of proofs: A new characterization of {NP}.
\newblock {\em Journal of the ACM}, 45(1):70--122, 1998.
\newblock Prelim. version FOCS '92.

\bibitem{BBT09}
N.~Bansal, S.~Bravyi, and B.~M. Terhal.
\newblock Classical approximation schemes for the ground-state energy of
  quantum and classical {I}sing spin {H}amiltonians on planar graphs.
\newblock {\em Quantum Information \& Computation}, 9(7\&8):0701--0720, 2009.

\bibitem{B31}
H.~Bethe.
\newblock Zur {T}heorie der {M}etalle.
\newblock {\em Zeitschrift f\"{u}r {P}hysik}, 71(3--4):205--226, 1931.

\bibitem{BH13}
F.~Brand{\~{a}}o and A.~Harrow.
\newblock Product-state approximations to quantum ground states.
\newblock In {\em Proceedings of the 45th {ACM} {S}ymposium on the {T}heory of
  {C}omputing (STOC 2013)}, pages 871--880, 2013.

\bibitem{B14}
S.~Bravyi.
\newblock Monte {C}arlo simulation of stoquastic {H}amiltonians.
\newblock {\em Quantum Information \& Computation}, 15(13\&14):1122--1140,
  2015.

\bibitem{BGKT18}
S.~Bravyi, D.~Gosset, R.~Koenig, and K.~Temme.
\newblock Approximation algorithms for quantum many-body problems.
\newblock Available at arXiv.org e-Print quant-ph/arXiv:1808.01734, 2018.

\bibitem{BH14}
S.~Bravyi and M.~Hastings.
\newblock On complexity of the quantum {I}sing model.
\newblock {\em Communications in Mathematical Physics}, 349(1):1--45, 2014.

\bibitem{BG16}
Sergey Bravyi and David Gosset.
\newblock Polynomial-time classical simulation of quantum ferromagnets.
\newblock {\em Physical Review Letters}, 119:100503, Sep 2017.
\newblock URL: \url{https://link.aps.org/doi/10.1103/PhysRevLett.119.100503},
  \href {http://dx.doi.org/10.1103/PhysRevLett.119.100503}
  {\path{doi:10.1103/PhysRevLett.119.100503}}.

\bibitem{BFV12}
J.~Bri\"{e}t, F.~M. de~Oliveira~Filho, and F.~Vallentin.
\newblock Grothendieck inequalities for semidefinite programs with rank
  constraint.
\newblock {\em Theory of Computing}, 10:77--105, 2014.

\bibitem{CW04}
Moses Charikar and Anthony Wirth.
\newblock Maximizing quadratic programs: Extending grothendieck's inequality.
\newblock In {\em 45th Annual IEEE Symposium on Foundations of Computer
  Science}, pages 54--60. IEEE, 2004.

\bibitem{CM16}
T.~Cubitt and A.~Montanaro.
\newblock Complexity classification of local {H}amiltonian problems.
\newblock {\em SIAM Journal on Computing}, 45(2):268--316, 2016.

\bibitem{F83}
U.~Fano.
\newblock Pairs of two-level systems.
\newblock {\em Reviews of Modern Physics}, 55:855--874, 1983.

\bibitem{GK11}
S.~Gharibian and J.~Kempe.
\newblock Approximation algorithms for {QMA}-complete problems.
\newblock {\em Siam Journal on Computing}, 41(4):1028--1050, 2012.

\bibitem{GK12}
S.~Gharibian and J.~Kempe.
\newblock Hardness of approximation for quantum problems.
\newblock In {\em Proceedings of 39th International Colloquium on Automata,
  Languages and Programming ({ICALP} 2012)}, pages 387--398, 2012.
\newblock DOI: 10.1007/978-3-642-31594-7, \copyright~2012 Springer,
  www.springerlink.com.

\bibitem{GHLS14}
Sevag Gharibian, Yichen Huang, Zeph Landau, and Seung~Woo Shin.
\newblock Quantum hamiltonian complexity.
\newblock {\em Foundations and Trends{\textregistered} in Theoretical Computer
  Science}, 10(3):159--282, 2014.
\newblock URL: \url{http://dx.doi.org/10.1561/0400000066}, \href
  {http://dx.doi.org/10.1561/0400000066} {\path{doi:10.1561/0400000066}}.

\bibitem{GW95}
M.~Goemans and D.~Williamson.
\newblock Improved approximation algorithms for maximum cut and satisfiability
  problems using semidefinite programming.
\newblock {\em Journal of the ACM}, 42:1115--1145, 1995.

\bibitem{H97}
D.~Hochbaum.
\newblock {\em Approximation Algorithms for {NP}-Hard Problems}.
\newblock Wadsworth Publishing Company, 1997.

\bibitem{HJ90}
R.~A. Horn and C.~H. Johnson.
\newblock {\em Matrix Analysis}.
\newblock Cambridge University Press, 1990.

\bibitem{HH96}
R.~Horodecki and M.~Horodecki.
\newblock Information-theoretic aspects of quantum inseparability of mixed
  states.
\newblock {\em Physical Review A}, 54(3):1838--1843, 1996.

\bibitem{HH96_2}
R.~Horodecki and P.~Horodecki.
\newblock Perfect correlations in the {E}instein-{P}odolsky-{R}osen experiment
  and {B}ell's inequalities.
\newblock {\em Physics Letters A}, 210:227, 1996.

\bibitem{KSV02}
A.~Kitaev, A.~Shen, and M.~Vyalyi.
\newblock {\em Classical and Quantum Computation}.
\newblock American Mathematical Society, 2002.

\bibitem{HL19}
E.~Lee and S.~Hallgren.
\newblock Approximation of {MAX}-2-local {H}amiltonians.
\newblock To be presented at the 19th Asian Quantum Information Science
  Conference (AQIS), 2019.

\bibitem{NV18}
A.~Natarajan and T.~Vidick.
\newblock Low-degree testing for quantum states, and a quantum entangled games
  pcp for qma.
\newblock In {\em Proceedings of the 59th IEEE Symposium on Foundations of
  Computer Science (FOCS)}, pages 731--742, 2018.

\bibitem{NC00}
M.~A. Nielsen and I.~L. Chuang.
\newblock {\em Quantum Computation and Quantum Information}.
\newblock Cambridge University Press, 2000.

\bibitem{O11}
T.~J. Osborne.
\newblock Hamiltonian complexity.
\newblock {\em Reports on Progress in Physics}, 75(2):022001, 2012.
\newblock URL: \url{http://stacks.iop.org/0034-4885/75/i=2/a=022001}.

\bibitem{PM15}
Stephen Piddock and Ashley Montanaro.
\newblock The complexity of antiferromagnetic interactions and 2d lattices.
\newblock {\em Quantum Information \& Computation}, 17(7-8):636--672, June
  2017.
\newblock URL: \url{http://dl.acm.org/citation.cfm?id=3179553.3179559}.

\bibitem{V01}
V.~Vazirani.
\newblock {\em Approximation Algorithms}.
\newblock Springer, 2001.

\end{thebibliography}

\appendix
\section{Max Cut as a special case of the Heisenberg model}\label{app:MC}

 We briefly sketch why local constraints $H_{ij}=I-Z_i\otimes Z_j$ in the Heisenberg model yield the NP-complete problem Max Cut. Namely, the Pauli $Z$ operator
 \[
    Z=\left(
        \begin{array}{cc}
          1 & 0 \\
          0 & -1 \\
        \end{array}
      \right)
 \]
 is diagonal in the standard basis with eigenvalues $1$ for eigenvector $\ket{0}$ and $-1$ for eigenvector $\ket{1}$. It follows that $Z\otimes Z$ also diagonalizes in the standard basis, with eigenvectors $\ket{00}$ and $\ket{11}$ attaining eigenvalue $1$ and $\ket{01}$ and $\ket{10}$ attaining eigenvalue $-1$. As a result, operator $I-Z_i\otimes Z_j$ has eigenvalues $0$ (with eigenspace spanned by $\ket{00}$ and $\ket{11}$) and $2$ (with eigenspace spanned by $\ket{01}$ and $\ket{10}$). But this means that on each edge $(i,j)\in E$,
 \begin{eqnarray*}
    \bra{00}I-Z_i\otimes Z_j\ket{00}&=&\bra{11}I-Z_i\otimes Z_j\ket{11}=0,\text{ and }\\
    \bra{01}I-Z_i\otimes Z_j\ket{01}&=&\bra{10}I-Z_i\otimes Z_j\ket{10}=2.
 \end{eqnarray*}
 In other words, if neighboring qubits are set to opposing standard basis states (e.g. $\ket{01}$), then we obtain value $2$ from an edge, and if the qubits are set to identical standard basis states (e.g. $\ket{00}$), we obtain value $0$ from this edge. Finally, since all local terms are diagonal in the standard basis, the entire Hamiltonian $H=\sum_{(i,j)\in E}H_{ij}$ will also be diagonal in the standard basis. The largest eigenvalue of $H$ will hence be the sum of the values obtained on each edge by the best standard basis state, which will correspond to a maximum cut in the graph. The actual largest eigenvalue will equal twice the maximum cut on the underlying graph (since we obtain value $2$ on each cut edge, rather than $1$ as for the standard Max Cut problem).

\section{Proofs for Section~\ref{scn:prelims}}\label{app:sec2proof}
\begin{proof}[Proof of Lemma~\ref{l:optprod}]
    Observe that for standard basis vectors $\ket{i},\ket{j},\ket{k},\ket{l}\in\complex^2$, we have
    \begin{eqnarray}
        \bra{ij}X\otimes X\ket{kl}&=&(i\oplus k)(j\oplus l),\label{eqn:1}\\
        \bra{ij}Y\otimes Y\ket{kl}&=&(-1)^{\delta_{kl}}(i\oplus k)(j\oplus l),\label{eqn:2}\\
        \bra{ij}Z\otimes Z\ket{kl}&=&(-1)^{k\oplus l}\delta_{ik}\delta_{jl}\label{eqn:3},
    \end{eqnarray}
    were $\delta_{ij}$ is the usual Kronecker delta. Denoting an arbitrary product state as $\ket{\psi}=ac\ket{00}+ad\ket{01}+bc\ket{10}+bd\ket{11}$ for $\abs{a}^2+\abs{b}^2=\abs{c}^2+\abs{d}^2=1$, we have
    \begin{eqnarray}
        \bra{\psi}H\ket{\psi}&=&\alpha(a^*c^*bd+acb^*d^*+a^*d^*bc+adb^*c^*)+\nonumber\\
        &&\beta(-a^*c^*bd-acb^*d^*+a^*d^*bc+adb^*c^*)+\nonumber\\
        &&\gamma(\abs{a}^2\abs{c}^2+\abs{b}^2\abs{d}^2-\abs{a}^2\abs{d}^2-\abs{b}^2\abs{c}^2)\label{eqn:exact}\\
        &=&2\operatorname{Re}[acb^*d^*](\alpha-\beta)+2\operatorname{Re}[adb^*c^*](\alpha+\beta)+\nonumber\\
        &&\gamma(\abs{a}^2\abs{c}^2+\abs{b}^2\abs{d}^2-\abs{a}^2\abs{d}^2-\abs{b}^2\abs{c}^2)\nonumber\\
        &\leq&2\abs{a}\abs{b}\abs{c}\abs{d}(\abs{\alpha+\beta}+\abs{\alpha-\beta})+\nonumber\\
        &&\abs{\gamma}\abs{(\abs{a}^2-\abs{b}^2)(\abs{c}^2-\abs{d}^2)}.\nonumber
    \end{eqnarray}
    where the last inequality follows from the triangle inequality. Let us simplify the notation above by assuming without loss of generality $a,b,c,d\in\reals^+$. We may also assume without loss of generality that $a\geq b$ and $c\geq d$ (since this maximizes the upper bound). Thus:
    \begin{eqnarray*}
        \bra{\psi}H\ket{\psi}&\leq&2{a}{b}{c}{d}(\abs{\alpha+\beta}+\abs{\alpha-\beta})+
        \abs{\gamma}({a}^2-{b}^2)({c}^2-{d}^2).
    \end{eqnarray*}
    Note now for any $\alpha,\beta\in\reals$, $\abs{\alpha+\beta}+\abs{\alpha-\beta}=\abs{\abs{\alpha}+\abs{\beta}}+\abs{\abs{\alpha}-\abs{\beta}}$. Assume first $\abs{\alpha}\geq \abs{\beta}$. Then
    \begin{equation}\label{eqn:bound2}
        \bra{\psi}H\ket{\psi}\leq4{a}{b}{c}{d}\abs{\alpha}+\abs{\gamma}({a}^2-{b}^2)({c}^2-{d}^2).
    \end{equation}
    Let $p=4abcd$ and $q=({a}^2-{b}^2)({c}^2-{d}^2)$. Note $p,q\geq 0$. Also, we claim $p+q\leq 1$; this will imply $\bra{\psi}H\ket{\psi}\leq \max(\abs{\alpha},\abs{\gamma})$. To see this claim, note
    \[
        p+q=(ac+bd)^2-(ad-bc)^2\leq(ac+bd)^2\leq 1,
    \]
    where the last inequality follows from the Cauchy-Schwarz inequality. The case of $\abs{\beta}\geq \abs{\alpha}$ follows analogously with $\abs{\alpha}$ in Equation~(\ref{eqn:bound2}) replaced with $\abs{\beta}$. We hence have $\bra{\psi}H\ket{\psi}\leq \max(\abs{\alpha},\abs{\beta},\abs{\gamma})=\snorm{(\abs{\alpha},\abs{\beta},\abs{\gamma})}$.

    We now show matching lower bounds, i.e. that $\abs{\alpha}$, $\abs{\beta}$, and $\abs{\gamma}$ are attainable. Returning to Equation~(\ref{eqn:exact}):
    \begin{itemize}
        \item For $\abs{\alpha}$: If $\alpha\geq 0$, set $a=b=c=d=1/\sqrt{2}$, and if $\alpha<0$, set $a=b=c=1/\sqrt{2}$ and $d=-1/\sqrt{2}$.
        \item For $\abs{\beta}$: If $\beta\geq 0$, set $a=i/\sqrt{2}$, $c=i/\sqrt{2}$, $b=d=1\sqrt{2}$, and if $\beta<0$, set $a=-i/\sqrt{2}$, $c=i/\sqrt{2}$, $b=d=1\sqrt{2}$.
        \item For $\abs{\gamma}$: If $\gamma\geq 0$, set $a=c=1$ and $b=d=0$. and of $\gamma<0$, set $a=d=1$, $b=c=0$.
    \end{itemize}
\end{proof}

\begin{proof}[Proof of Lemma~\ref{l:opt}]
    Denoting an arbitrary two-qubit state as $\ket{\psi}=a\ket{00}+b\ket{01}+c\ket{10}+d\ket{11}$ for $\abs{a}^2+\abs{b}^2=\abs{c}^2+\abs{d}^2=1$, we have via Equations~(\ref{eqn:1})-(\ref{eqn:3}) that
    \begin{eqnarray*}
        \bra{\psi}X\otimes X\ket{\psi}&=&a^*d+ad^*+b^*c+bc^*,\\
        \bra{\psi}Y\otimes Y\ket{\psi}&=&-a^*d-ad^*+b^*c+bc^*,\\
        \bra{\psi}Z\otimes Z\ket{\psi}&=&\abs{a}^2-\abs{b}^2-\abs{c}^2+\abs{d}^2.
    \end{eqnarray*}
    Thus, $\bra{\psi}H\ket{\psi}$ equals
    \begin{eqnarray}
        &&\alpha(2\operatorname{Re}(ad^*)+2\operatorname{Re}(bc^*))+
        \beta(-2\operatorname{Re}(ad^*)+2\operatorname{Re}(bc^*))+
        \gamma(\abs{a}^2+\abs{d}^2-\abs{b}^2-\abs{c}^2)\nonumber\\
        &=&2\operatorname{Re}[ad^*](\alpha-\beta)+2\operatorname{Re}[bc^*](\alpha+\beta)
        +(\abs{a}^2+\abs{d}^2-\abs{b}^2-\abs{c}^2)\gamma.\nonumber
    \end{eqnarray}
    Observe that since the coefficient of $\gamma$ depends on only \emph{absolute values} of $a,b,c,d$, we can assume without loss of generality that the optimal assignment has $a,b,c,d\geq 0$ and satisfies
    \[
        \bra{\psi}H\ket{\psi}=2ad\abs{\alpha-\beta}+2bc\abs{\alpha+\beta}
        +({a}^2+{d}^2-{b}^2-{c}^2)\gamma.
    \]
    By applying the Arithmetic-Geometric mean inequality, we hence have
    \begin{eqnarray*}
        \bra{\psi}H\ket{\psi}&\leq& (a^2+d^2)\abs{\alpha-\beta}+(b^2+c^2)\abs{\alpha+\beta}+(a^2-b^2-c^2+d^2)\gamma\\
        &=&(a^2+d^2)(\abs{\alpha-\beta}+\gamma)+(b^2+c^2)(\abs{\alpha+\beta}-\gamma)\\
        &\leq&\max(\abs{\alpha-\beta}+\gamma,\abs{\alpha+\beta}-\gamma),
    \end{eqnarray*}
    where the last statement follows since $a^2+b^2+c^2+d^2=1$. The matching lower bound is obtained as follows. To achieve $\abs{\alpha-\beta}+\gamma$ when $\alpha\geq \beta$, set $a=d=1/\sqrt{2}$, and when $\alpha\leq \beta$, set $a=1/\sqrt{2},d=-1/\sqrt{2}$. Similarly, to achieve $\abs{\alpha+\beta}-\gamma$ when $\alpha\geq -\beta$, set $b=c=1/\sqrt{2}$, and when $\alpha\leq -\beta$, set $b=1/\sqrt{2},c=-1/\sqrt{2}$.

\end{proof}

\section{Lemmas and Mathematica code}\label{app:lemmas}
In Section~\ref{sscn:tight} we use the following lemma, which is stated as given in~\cite{BFV12}. Below, $_2F_1(a,b;c;z)$ is the hypergeometric function, defined for $\abs{z}<1$ as
\[
    _2F_1(a,b;c;z)=\sum_{n=0}^\infty\frac{(a)_n(b)_n}{(c)_n}\frac{z^n}{n!},
\]
where for $n\geq 0$, we have Pochhammer symbol $(x)_n=\Gamma(x+n)/\Gamma(x)=x(x+1)\cdots(x+n-1)$ for $\Gamma$ the Gamma function.

\begin{lemma}[Bri\"{e}t, de Oliveira Filho and Vallentin~\cite{BFV12}]\label{l:BFV}
    Let $u,v$ be unit vectors in $\reals^n$ and let $Z\in\reals^{r\times n}$ be a random matrix whose entries are distributed independently according to the standard normal distribution with mean $0$ and variance $1$. Then,
    \begin{equation}\label{eqn:BFV}
        \mathds{E}\left[\frac{Zu}{\enorm{Zu}}\cdot\frac{Zv}{\enorm{Zv}}\right]=\frac{2}{r}\left(\frac{\Gamma((r+1)/2)}{\Gamma(r/2)}\right)^2(u\cdot v)\;_2F_1\left(1/2,1/2;r/2+1;(u\cdot v)^2\right).
    \end{equation}
\end{lemma}

\paragraph*{Mathematica code.} Below, we give the Mathematica code used to numerically calculate the approximation ratios of Theorem~\ref{thm:alg1}:\\

\texttt{
\begin{tabular}{l}
  g[r\_] := 2/r (Gamma[(r + 1)/2]/Gamma[r/2])\textasciicircum 2 \\
  F[r\_, t\_] := g[r] t Hypergeometric2F1[1/2, 1/2, r/2 + 1, t\textasciicircum 2] \\
  ApproxRatio[r\_] :=
Min[Select[
Table[(1 - F[r, t])/(1 - r t),\\
 \hspace{76mm}\{t, -1, 1/r, 0.01\}], \# > 0 \&]] \\
  ApproxRatio[1] \\
  ApproxRatio[2] \\
  ApproxRatio[3] \\
\end{tabular}
}\\

\noindent The code for the approximation ratios in Section~\ref{ssscn:varyweights} is:\\

\texttt{
\begin{tabular}{l}
ApproxRatio[r\_] :=\\
  Min[Select[
    Flatten[Table[(1 - g 2 ArcSin[t]/Pi)/(1 - r g t), \\
    \hspace{55mm} \{g, -1, 1\}, \{t, -1, 1, 0.01\}]], \# > 0 \&]] \\
ApproxRatio[1]\\
ApproxRatio[2]\\
ApproxRatio[3]\\
\end{tabular}
}

\end{document}